\documentclass{article}
\usepackage{amsmath,amssymb, amsthm, graphicx}
\newtheorem{lemma}{Lemma}
\newtheorem{prop}{Proposition}
\newtheorem{thm}{Theorem}
\newtheorem{coro}{Corollary}

\theoremstyle{remark}
\newtheorem{ex}{Example}
\newtheorem{rem}{Remark}

\def\Ce {\mathcal C}
\def\Ha {\mathcal H}
\def\Ka {\mathcal K}

\def\Ve {\mathcal V}
\def\De {\mathcal D}
\def\Me {\mathcal M}

\def\Oe {\mathcal O}
\def\Te{\mathcal T}
\def\Pe {\mathcal P}
\def\Le {\mathcal L}
\def\ee {\mathcal E}

\def\Se {\mathfrak S}

\def\<{\langle}
\def\>{\rangle}
\def\Tr{{\rm Tr}\,}

\def\b{\{b\}}
\def\supp{{\rm supp}\,}

\def\qed{{\hfill $\square$}\vskip 5pt}

\title{Base norms  and  discrimination of generalized quantum channels}
\author{A. Jen\v cov\' a\\
Mathematical Institute, Slovak Academy of Sciences\\
\v Stef\'anikova 49, Bratislava, Slovakia}
\date{}
\begin{document}
\maketitle
{\abstract We introduce and study norms in the space of hermitian matrices, obtained from base norms in positively generated
 subspaces. These norms are closely related to discrimination of so-called generalized quantum channels, including quantum
 states, channels and networks. We further introduce generalized quantum decision problems and show that the maximal average payoff of decision procedures is again given by these norms. We also study optimality of  decision procedures, 
in particular, we obtain a necessary and sufficient condition
under which an optimal 1-tester for dicrimination of quantum channels exists, such that the input state is maximally entangled. }
\section{Introduction and preliminaries}

It is well known that in the problem of discrimination of quantum states, the best possible distinguishability of
two states $\rho_0$ and $\rho_1$ is given by the trace norm $\| \rho_0-\rho_1\|_1$, \cite{helstrom, holevo}.
 The set of states forms a base of the convex cone of positive operators and
the restriction of the trace norm  to hermitian operators is the corresponding base  norm.
Similarly,  it was shown in \cite{wolf} that more
 general distinguishability measures, obtained by specification of the allowed measurements e.g. for bipartite states,
 are obtained from base norms associated with more general positive cones. This correspondence is related to  duality
 of  the base norm  and  the order unit norm, with respect to a given positive cone.

  In a similar problem for quantum channels, and recently also quantum networks, the diamond norm $\|\cdot\|_\diamond$ for channels \cite{kitaev}, resp. the strategy $N$-norm 
$\|\cdot\|_{N\diamond}$ \cite{gutoski, daria_testers} for networks is
 obtained. Via the Choi isomorphism, quantum networks are represented by certain positive operators on the tensor
 product of the input and output spaces, so-called $N$-combs \cite{daria_circ, daria}, see also \cite{guwat}. The set of $N$-combs is the intersection 
of the multipartite state space by a positively generated subspace of the real vector space of hermitian operators.
Since this subspace inherits the order structure and the set of $N$-combs forms a base of its
 positive cone, it is natural to expect that the distinguishability norm $\|\cdot\|_{N\diamond}$ is in fact  the corresponding 
base norm.

Motivated by this question, we study positively generated subspaces of the space of hermitian operators 
$B_h(\Ha)$ acting on a finite dimensional Hilbert space $\Ha$. For a 
given base $B$ of the positive cone, we define a distinguishability   measure  in terms of tests that are defined as affine maps 
$B\to[0,1]$ and show that this measure is given by the base norm, this, in fact, is easy to see for any finite dimensional
 ordered vector space. We then study a natural extension of this norm to $B_h(\Ha)$ and its dual norm. An example of 
 such a  base  is the set of Choi matrices of so-called generalized channels, 
this contains the set of $N$-combs as a special case.  For $N$-combs,
 the obtained norm coincides with $\|\cdot\|_{N\diamond }$ and we  recover some of the results of \cite{gutoski}
 concerning the dual norm. Moreover, we find a suitabe expression for this norm, closely related to the definition 
 of $\|\cdot\|_\diamond$. 
 
 It the last section, we introduce generalized quantum decision problems with respect to a base $B$. We show that maximal average payoff (or minimal average loss) of generalized decision procedures is again given by a  base norm. We find optimality conditions for generalized decision procedures, in particular, for quantum measurements 
 and testers. In the case of multiple
hypothesis testing for states,
we get the results  obtained previously in \cite{renner,damo,holevo1}. In the case of discrimination of quantum channels,
we find a necessary and sufficient condition for existence of an optimal 
 tester such that the input state is maximally entangled.

The rest of the present section contains some basic definitions and preliminary results on discrimination of quantum devices, as well as convex cones, bases and base norms. 

\subsection{Discrimination of quantum states, channels and networks}

Let $\Ha$ be a finite dimensional Hilbert space and let $B(\Ha)$ be the set of bounded operators on $\Ha$.
We denote by $B_h(\Ha)$ the set of self-adjoint operators,  $B(\Ha)^+$ the  cone of positive operators 
and 
 $\Se(\Ha):=\{ \rho\ge 0, \Tr \rho=1\}$  the set of states in $B(\Ha)$. Let $\Ka$ be  another finite dimensional 
Hilbert space. It is well known that $B(\Ka\otimes \Ha)$ corresponds to 
the  set of all linear maps $B(\Ha)\to B(\Ka)$, via the Choi representation:
\begin{equation}\label{eq:choirep}
X_\Phi=(\Phi\otimes id_\Ha)(\Psi),\qquad \Phi_X(a)=\Tr_\Ha[(I_\Ka\otimes a^{\mathsf T})X]
\end{equation}
here $\Psi=|\psi\>\<\psi|$ and $|\psi\>=\sum_i |i\>\otimes |i\>$ for an ONB $\{|i\>, i=1,\dots,\dim(\Ha)\}$ in $\Ha$,
 $a^\mathsf T$ denotes transpose of $a$. 
In this correspondence, $B(\Ka\otimes\Ha)^+$ is identified with the set of completely positive maps and $B_h(\Ka\otimes\Ha)$ 
with hermitian maps, that is, maps satisfying $\Phi(a^*)=\Phi(a)^*$.

 Consider the problem of quantum state discrimination: suppose the quantum system represented by $\Ha$ is
 known to be in one of two given states $\rho_0$ or $\rho_1$ and the task is to decide which of them is the true state.
This is done by using a test, that is  a binary positive operator valued measure (POVM). This is given by an operator
 $0\le M\le I$, with the interpretation that 
$\Tr M\rho$ is the probability of deciding for $\rho_0$ if the true value of the state is $\rho$. Equivalently, a test 
 can be defined as an affine map $\Se(\Ha)\to [0,1]$. 

Given an a priori probability $0\le \lambda\le1$ that the true state is $\rho_0$,  we need to minimize the average
 probability of error over all tests, that is to find the value of
 \[
\Pi_\lambda(\rho_0,\rho_1):=\min_{0\le M\le I}\lambda \Tr (I-M)\rho_0+(1-\lambda)\Tr M\rho_1,
 \]
this is the minimum Bayes error probability. Then \cite{helstrom, holevo} 
\[
\Pi_\lambda(\rho_0,\rho_1)=\frac12-\frac12\|\lambda\rho_0-(1-\lambda)\rho_1\|_1,
\]
where $\|a\|_1:=\Tr |a|$, $a\in B(\Ha)$ is the trace norm.

Let now $\Ha$ and $\Ka$ be two finite dimensional Hilbert spaces and consider the problem of discrimination 
of channels. Here we have to decide between two channels
$\Phi_0$ and $\Phi_1$ and this time the tests are given by binary quantum 1-testers \cite{daria_testers}, or PPOVMs \cite{ziman},
 which are  positive operators $T\in B(\Ka\otimes \Ha)^+$, such that $T\le I_\Ka\otimes \sigma$ for some
 $\sigma\in \Se(\Ha)$.
 These correspond to triples $(\Ha_A,\rho,M)$, where $\Ha_A$ is an ancillary Hilbert space, $\rho\in \Se(\Ha\otimes\Ha_A)$
 and $0\le M\le I$, $M\in B(\Ka\otimes\Ha_A)$.  The probability of choosing $\Phi_0$ if the true value is $\Phi$
 for a tester $T$ is given by 
\[
p(T,\Phi):=\Tr T X_\Phi=\Tr M(\Phi\otimes id_A)(\rho).
\]
The minimum Bayes error probability is now
\[
\Pi^1_\lambda(\Phi_0,\Phi_1):=\min_T \lambda(1- p(T,\Phi_0))+(1-\lambda)p(T,\Phi)=\frac12-\frac12\|\lambda\Phi_0-(1-\lambda)\Phi_1\|_\diamond
\]
where the diamond norm $\|\Phi\|_\diamond$ for a hermitian map $\Phi$ is defined as
\cite{kitaev, watrous}
\begin{eqnarray*}
\|\Phi\|_\diamond&=&\sup_{\dim(\Le')<\infty}\, \sup_{\rho\in\Se(\Ha\otimes\Le')}\|\Phi\otimes id_{\Le'}(\rho)\|_1\\
&=&\sup_{\rho\in\Se(\Ha\otimes\Le)}\|\Phi\otimes id_{\Le}(\rho)\|_1,\qquad \dim(\Le)=\dim(\Ha)
\end{eqnarray*}
By duality, this norm is related to the $cb$-norm for completely bounded linear maps, see \cite{paulsen}.
\begin{figure}
\begin{center}
\includegraphics{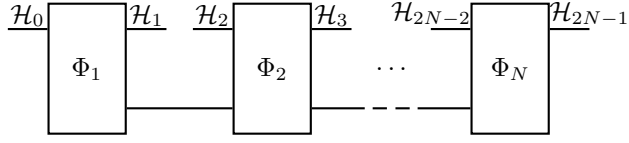}
\caption{A deterministic quantum $N$-comb}\label{fig:N_comb}

\end{center}
\end{figure}

Let now $\{\Ha_0$, $\Ha_1,\dots, \Ha_{2N-1}\}$ be finite dimensional Hilbert spaces. Consider a sequence of channels
$\Phi_i: B(\Ha_{2i-2}\otimes\Ha_A)\to B(\Ha_{2i-1}\otimes \Ha_A)$, $i=1,\dots,N$, connected by the ancilla $\Ha_A$
 as indicated on fig. \ref{fig:N_comb} (the first and last ancilla is traced out).
 This defines a channel
$\Phi: B(\Ha_0\otimes\Ha_2\otimes\dots\otimes \Ha_{2N-2})\to B(\Ha_1\otimes \Ha_3\otimes\dots\otimes\Ha_{2N-1})$, such
 channels describe quantum networks. The channels $\Phi_1,\dots,\Phi_N$ are not unique, in fact, these can always be 
supposed to be isometries.
 A (deterministic) quantum $N$-comb is defined as the Choi matrix  $X_\Phi$ of such a channel, see \cite{daria} for more about quantum networks and $N$-combs. The same definition, called a (non-measuring) quantum $N$-round strategy, 
was also introduced in 
\cite{guwat}. A (non-measuring) quantum $N$-round  co-strategy can be  defined as an $(N+1)$-strategy for the sequence of
 spaces $\{\mathbb C,\Ha_0,\dots,\Ha_{2N-1},\mathbb C\}$. 

The tests for discrimination of two networks $\Phi^0$ and $\Phi^1$ are given by quantum $N$-testers, which are obtained 
 by an $(N+1)$-comb such that the first channel has 1-dimensional input space (hence is a state) and  a
(binary) POVM is applied to the ancilla \cite{daria, daria_testers}, see figs. \ref{fig:tester}, \ref{fig:apply}. 
This can be represented by a pair $(T_0,T_1)$ of positive operators, such that $T_0+T_1$ is an $(N+1)$-round co-strategy,
 \cite{daria, guwat,gutoski}.
\begin{figure}
\begin{center}
\includegraphics{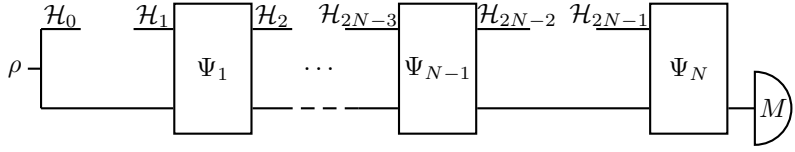}
\caption{A quantum $N$-tester}\label{fig:tester}

\end{center}
\end{figure}

\begin{figure}
\begin{center}
\includegraphics{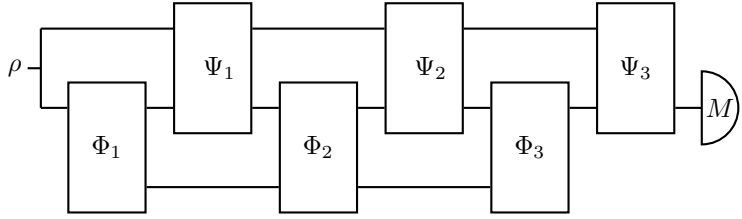}
\caption{A 3-tester applied to a 3-comb}\label{fig:apply}

\end{center}
\end{figure}

The minimal Bayes error probability now has the form
\[
\Pi^N_\lambda(\Phi^0,\Phi^1)=\frac12-\frac12\|\lambda \Phi^0-(1-\lambda)\Phi^1\|_{N\diamond}
\]
where for any hermitian $\Phi$,
\begin{eqnarray}
\|\Phi\|_{N\diamond}&=& \sup_T\|(T_0+T_1)^{1/2}X_\Phi(T_0+T_1)^{1/2}]\|_1,\qquad \cite{daria_testers}\\
&=& \sup_T\Tr X_\Phi(T_0-T_1),\qquad \cite{gutoski}
\end{eqnarray}
where the supremum is taken over all $N$-testers. The dual norm was also obtained in \cite{gutoski}, as
\[
\|\Phi\|_{N\diamond}^*=\sup_S\Tr X_\Phi(S_1-S_0),
\]
where the supremum is taken over the set of pairs of positive operators such that $S_0+S_1$ is an $N$-round strategy
 ($N$-comb).

\subsection{Convex cones, bases and base norms}\label{sec:basen}

Let $\Ve$ be a finite dimensional real vector space and let $\Ve^*$ be the dual 
space, with duality $\<\cdot,\cdot\>$. A subset  $Q\subset \Ve$ is a 
{convex cone}  if $\lambda q_1+ \mu q_2\in Q$ whenever $q_1,q_2\in Q$ and 
$\lambda,\mu\ge 0$.  The cone is 
pointed if $Q\cap-Q=\{0\}$ and generating if $\Ve=Q-Q$. Closed pointed convex cones are in one-to-one correspondence with partial orders in $\Ve$, by $x\le_Q y \iff y-x\in Q$.  

The {dual cone} of $Q$ is defined as
\[
Q^*=\{f\in \Ve^*, \<f,q\>\ge 0, q\in Q\}
\]
This is a closed convex cone and $Q^{**}=Q$ if $Q$ is closed. Moreover, a closed convex cone $Q$ is pointed if and only if $Q^*$ is generating . A closed  pointed generating convex cone is called a proper cone.

A {base} of the proper cone $Q$ is a compact convex subset $B\subset Q$, such that each nonzero element $q\in Q$ has a unique representation in the form $q=\lambda b$ with $b\in B$ and $\lambda>0$. It is clear that any base {generates} the cone $Q$, in the sense that $Q=\bigcup_{\lambda\ge 0}\lambda B$.
Then any element $v\in \Ve$ can be written as $v=\lambda b_1-\mu b_2$, $\lambda,\mu\ge 0$, $b_1,b_2\in B$.

For any base $B$, the map $Q\ni q=\lambda b\mapsto \lambda$ extends uniquely to a linear functional $e_B\in Q^*$ and we have
$B=\{q\in Q, 
\<e_B,q\>=1\}$.

\begin{lemma} Let $f\in Q^*$. Then $f\in int(Q^*)$ if and only if
\[
B_{f}:=\{q\in Q, \<f,q\>=1\}
\]
is a base of $Q$.
\end{lemma}

{\it Proof.} It is quite clear that $B_{f}$ is a base of $Q$ if and only if
$\<f,q\>>0$ for any nonzero $q\in Q$. By \cite[Theorem 11.6]{rockafellar},
 this is equivalent with $f\in int (Q^*)$.

\qed

Let $\le$ denote the order in $\Ve$ given by $Q$.
An element $e\in \Ve$ is an {order unit} in $\Ve$ if for any 
$v\in \Ve$, there is some
 $r>0$ such that $re\ge v$. It is easy to see that $e$ is an order unit if and only if $e\in int(Q)$. Consequently,

\begin{coro}\label{coro:base} Any base $B$ of $Q$ defines an order unit $e_B$ in  $\Ve^*$ and, 
conversely, any order unit $e$ in  $\Ve^*$ defines a base $B_e$ of $Q$. We have 
$e_{B_e}=e$ and $B_{e_B}=B$. 

\end{coro}

Let $B$ be a base of $Q$. The corresponding {base norm} in $\Ve$ is defined by
\[
\|v\|_B= \inf\{\lambda+\mu,\ v=\lambda b_1-\mu b_2, \lambda,\mu\ge 0,\ b_1,b_2\in B\}
\]
It is clear that $\|q\|_B=\<e_B,q\>$ for all $q\in Q$.
Let $\Ve_1$ be the unit ball of $\|\cdot\|_B$ in $\Ve$, then 
\[
\Ve_1=\{\lambda b_1-\mu b_2, \ b_1,b_2\in B, \lambda,\mu\ge 0, \lambda+\mu=1\}=co(B\cup -B)
\]
Let $\|\cdot\|_B^*$ be the dual norm in $\Ve^*$, then the unit ball $\Ve^*_1$ for $\|\cdot\|_B^*$  is given by
\[
\Ve_1^*=\Ve_1^\circ=(co(B\cup -B))^\circ=(B\cup -B)^\circ=B^\circ\cap (-B)^\circ
\]
where $A^\circ:=\{f\in \Ve^*,\ \<f,a\>\le 1, a\in A\}$ is the polar of  
$A\subset \Ve$, see \cite{rockafellar}. We have
\[
\Ve_1^*=\{f\in \Ve^*, -1\le \<f,b\>\le 1, b\in B\}=\{f\in \Ve^*,
 -e_B\le_{Q^*} f\le_{Q^*} e_B\}
\]
where $e_B$ is the order unit. Hence  the dual norm is given by
\[
\|f\|_B^*=\inf \{\lambda>0, -\lambda e_B\le_{Q^*} f\le_{Q^*} \lambda e_B\}=:\|f\|_{e_B}
\] 
In general, if  $e$ is an order unit, then $\|\cdot\|_e$ defines a norm  called
the {order unit norm} in $\Ve^*$. 

Since $\|\cdot\|_B$ is the dual norm for $\|\cdot\|_{e_B}$, we get for  $v\in \Ve$,
\begin{equation}\label{eq:basenorm}
\|v\|_B=\|v\|_{e_B}^*=\sup_{-e_B\le_{Q^*} f\le_{Q^*} e_B}\<f,v\>=
2\sup_{f\in Q^*, f \le_{Q^*} e_B}\<f,v\>-\<e_B,v\>
\end{equation}
where the last equality follows by replacing $f$ by $\tfrac12 (f+e_B)$. 
\begin{ex}\label{ex:states}
Let  $\Ve=B_h(\Ha)$ be the real vector space of self-adjoint elements in $B(\Ha)$ and
let $Q=B(\Ha)^+$. We identify $\Ve^*$ with $\Ve$, with duality 
$\<a,b\>=\Tr ab$, then $Q$ is a self-dual proper cone and $B=\Se(\Ha)$ is a base of $Q$, with $e_B=I$. The order unit norm
 $\|\cdot\|_{I}$ is the operator norm $\|\cdot\|$ in $B(\Ha)$ and its dual $\|\cdot \|_B$ is  the trace norm $\|\cdot\|_1$.

\end{ex}

We will finish this section by showing that the base norm is naturally  related to a distinguishability measure for
 elements of the base. By analogy with the set of quantum states, let us define a test on a base 
$B$ as an affine map 
$\mathbf t:B\to [0,1]$.
 It is easy to see that there is a one-to-one correspondence between tests on $B$ and elements  $f\le_{Q^*} e_B$ 
 in $Q^*$. Let $b_0,b_1$ be two elements of $B$ and let us interpret the value $\mathbf t(b)=\<f,b\>$
 as the probability of choosing $b_0$  if the ''true value'' is $b$. Then $\<f,b_1\>$ and $1-\<f,b_0\>$ are probabilities
of making an error. Let 
$\lambda\ge 0$, then we define the minimal average error probability as
 \[
\Pi^B_{\lambda}(b_0,b_1):=\min_{0\le_{Q^*} f\le_{Q^*} e_B} \lambda (1-\<f,b_0\>)+(1-\lambda)\<f,b_1\>
 \]
We obtain by (\ref{eq:basenorm}) that
\begin{eqnarray*}
\Pi^B_{\lambda}(b_0,b_1)&=&\lambda-\max_{0\le_{Q^*}f\le_{Q^*} e_B}\<f,\lambda b_0-(1-\lambda)b_1\>\\
&=& \frac 12(1 - \|\lambda b_0-(1-\lambda)b_1\|_B).
\end{eqnarray*}

\section{Base norms on subspaces of $B_h(\Ha)$}\label{sec:subops}

We now put $\Ve=B_h(\Ha)$, with the self-dual proper cone $B(\Ha)^+$ as in Example \ref{ex:states}. We will describe all possible bases of this cone.

It is clear that $int(B(\Ha)^+)$ is the set of positive definite elements, hence these are the order units in 
$B_h(\Ha)$.
By Corollary \ref{coro:base}, there is a one-to-one 
correspondence between  positive definite elements $b$ and bases  of $B(\Ha)^+$, given by
\begin{equation}\label{eq:bases}
b\leftrightarrow S_b:=\{a\in B(\Ha)^+, \Tr ab=1\}=B(\Ha)^+\cap \Te_b,
\end{equation}
where $\Te_b=\{x\in B_h(\Ha), \Tr xb=1\}$.
  For $b\in int(B(\Ha)^+)$, we have by (\ref{eq:basenorm}) and Example \ref{ex:states} that the corresponding base norm is
\begin{equation}\label{eq:basenorm_b}
\|x\|_{S_b}=\sup_{-b\le a\le b}\Tr ax=\sup_{-I\le a\le I} \Tr ab^{1/2}xb^{1/2}=\|b^{1/2}xb^{1/2}\|_1
\end{equation}
and the dual order unit norm is
\begin{equation}\label{eq:orderunitnorm_b}
\|x\|_b=\inf\{\lambda>0, -\lambda b\le x\le\lambda b\}=\|b^{-1/2}xb^{-1/2}\|.
\end{equation}
If $b\in B(\Ha)^+$ is any element, we define
\[
\|b^{-1/2}xb^{-1/2}\|:=\lim_{\varepsilon\to 0^+}\|(b+\varepsilon)^{-1/2}x(b+\varepsilon)^{-1/2}\|
\]
Note that the expression on the RHS is bounded for all $\varepsilon>0$ if and only if $\supp(x)\le \supp(b)$ and in this 
case the norm on the LHS is defined by restriction to the support of $b$. Otherwise, the limit is infinite.
Moreover, for any elements $a,b\in B(\Ha)^+$, we define
\[
D_{max}(a\| b):=\log\inf\{\lambda>0, a\le \lambda b\}=\inf\{\gamma>0, a\le 2^\gamma b\}
\]
For a pair of states $\rho$ and $\sigma$,  $D_{max}(\rho\| \sigma)$ is the
max-relative entropy of $\rho$ and $\sigma$, introduced in \cite{datta}.\footnote{Note that $D_{max}$ was denoted by $D_\infty$ in \cite{renner}.} If $b\in int(B(\Ha)^+)$, then
\[
D_{max}(a\| b)=\log (\|a\|_b)
\]
In general, if $\supp( a)\le \supp (b)$, then we may restrict to the support of $b$ and with this restriction
 $D_{max}(a\| b)=\log(\|a\|_b)$,
 otherwise $D_{max}(a\| b)=\infty$.

\subsection{Sections of a base of $B(\Ha)^+$}

Let $J\subset B_h(\Ha)$ be a  subspace and let $Q=J\cap B(\Ha)^+$ be the convex cone of positive elements in $J$.
It is obvious that $Q$ is closed and pointed. We will   suppose that $J$ is positively generated, then $J=Q-Q$ and 
   $Q$ is a proper cone in $J$. Let $b\in Q$ be 
such that $\supp a\le \supp b=:p$ for all $a\in Q$, then $J\subseteq B_h(p\Ha)$ and by restricting to $B_h(p\Ha)$, 
we may suppose that $b$ is positive definite. Conversely, if $J$ contains a positive definite element, then 
$J$ is positively generated.

Let $J^\perp=\{ y\in B_h(\Ha), \Tr xy=0,\ x\in J\}$, let $B_h(\Ha)|_{J^\perp}$ be the quotient space and let 
$\pi: B_h(\Ha)\to B_h(\Ha)|_{J^\perp}$ be
the quotient map $a\mapsto a+J^\perp$. We may identify the dual space $J^*$ with  
$B_h(\Ha)|_{J^\perp}$, with duality 
\[
\<x,\pi(a)\>=\Tr xa, \qquad x\in J,\ a\in B_h(\Ha).
\]
It was shown in \cite{ja} 
that the dual cone  of $Q$ is $Q^*=\pi(B(\Ha)^+)$, moreover, since $\pi$ is a linear map, we have 
$int(Q^*)=int(\pi(B(\Ha)^+))=\pi(int(B(\Ha)^+))$ by \cite{rockafellar}. In other words, any element $f\in Q^*$ has the form
\[
f(x)=\Tr ax,\qquad x\in J
\]
for some (in general non-unique) element $a\in B(\Ha)^+$ and $f$ is an order unit in $J^*$ if and only if $a$ may be chosen positive definite.   Now we can use Corollary \ref{coro:base} to describe all bases of $Q$.

\begin{lemma}\label{lemma:bases_of_Q} 
A subset $B\subset Q$ is a base of $Q$ if and only if  $B=J\cap S_{\tilde b}$, where $\tilde b\in int(B(\Ha)^+)$.
 In this case,  $\pi(\tilde b)=e_B$.

\end{lemma}

\begin{proof} Let $B$ be a base of $Q$. Since $e_B\in int(Q^*)$, there is some $\tilde b\in int(B(\Ha)^+)$ such that
$e_B=\pi(\tilde b)$ and 
\[
B=\{q\in Q, \Tr q\tilde b=\<e_B,q\>=1\>\}=Q\cap \Te_{\tilde b}=J\cap S_{\tilde b}
\]
(see (\ref{eq:bases})).
 Conversely, it is quite clear that $B=J\cap S_{\tilde b}$  is a base of $Q$ and $e_B=\pi(\tilde b)$.
 
\end{proof}

A set of the form $B=L\cap S_{\tilde b}$ where $\tilde b\in int(B(\Ha)^+)$ and $L\subseteq B_h(\Ha)$ is a subspace will be called a 
  section of a base of $B(\Ha)^+$, or simply a section. Let $J=\mathrm{span}(B)$ be the real linear span of $B$, then 
\[
B\subseteq J\cap S_{\tilde b}\subseteq L\cap S_{\tilde b}=B,
\]
so that $B=J\cap S_{\tilde b}$ and $B$ is a base of $Q=J\cap B(\Ha)^+$.
If moreover $B$  contains a positive definite element, 
we say that $B$ is a faithful section. In this case, we have 
$B\cap int(B(\Ha)^+)=ri(B)$, where $ri(B)$ denotes the relative interior  of $B$. Indeed, since 
$B=L_{\tilde b}\cap B(\Ha)^+$, where $L_{\tilde b}=:L\cap \Te_{\tilde b}$ is an affine subspace and
$L_{\tilde b}\cap int(B(\Ha)^+)\ne \emptyset$, we have by \cite[Theorem 6.5]{rockafellar} that
\[
ri(B)= ri(L_{\tilde b})\cap ri(B(\Ha)^+)=L_{\tilde b}\cap int(B(\Ha)^+)=B\cap int(B(\Ha)^+)
\]
For example, note that if $B=\{b\}$ for some $b\in B(\Ha)^+$, then $B$ is a section and $B$ is  faithful  if and only if $b$  is positive definite. If a section $B$ is not faithful, then there is some element $b\in B$ such that  $p=\supp (b)$ and $B\subset B(p\Ha)$. Then $B$ is 
a faithful section of a base of $B(p\Ha)^+$, in this case, $ri(B)=B\cap ri(B(p\Ha)^+)$. 
From now on, we will suppose that  $B$ is a faithful section and we put 
$J:=\mathrm{span}(B)$, $Q:=J\cap B(\Ha)^+$.

Note that in Lemma \ref{lemma:bases_of_Q},  the correspondence between the base $B$ and the element $\tilde b$ 
 such that $B=J\cap S_{\tilde b}$ is not one-to-one, since the order 
unit $e_B=\pi(\tilde b)$  may  contain more different positive definite elements.
We will now look at the set of all such elements.
 Let  
 \[
 \tilde B:=\{\tilde b\in B(\Ha)^+, \Tr b\tilde b=1, b\in B \}
 \]
 Note that
\begin{equation}\label{eq:tildeb} 
\tilde B= \pi^{-1}(e_B)\cap B(\Ha)^+=(\tilde b+J^\perp)\cap B(\Ha)^+, 
\end{equation}
where $\tilde b$ is any element in $\tilde B$.  Let $\tilde J:=\mathrm{span}(\tilde B)$ and 
$\tilde Q:=\tilde J\cap B(\Ha)^+$. Note that
 $\tilde B$ always contains a positive definite element, so that $\tilde J$  is positively generated and  $\tilde Q$ is a  proper cone in $\tilde J$. Since by (\ref{eq:tildeb}) $\tilde B$ is an intersection of $B(\Ha)^+$ by an affine 
subspace, we have 
\[
\{\tilde b\in int(B(\Ha)^+),\ B=J\cap S_{\tilde b}\}=\tilde B\cap int(B(\Ha)^+)=ri(\tilde B).
\]

\begin{lemma}\label{lemma:B_tilde_B} Let $b\in ri(B)$, $\tilde b\in ri(\tilde B)$. Then
\begin{enumerate}

\item[(i)]  $\tilde B=\tilde J\cap S_b$, so that $\tilde B$ is a faithful section of a base of $B(\Ha)^+$.

 \item[(ii)]  $\tilde{\tilde B}=B$.

\item [(iii)] $B=\bigcap_{\tilde b'\in ri(\tilde B)} S_{\tilde b'}$.

\end{enumerate}
\end{lemma}

\begin{proof} (i) 
Since $\tilde B$ is  convex, any element $y\in \tilde J$ has the form
$y=\lambda \tilde b_1-\mu \tilde b_2$, with $\tilde b_1,\tilde b_2\in \tilde B$ and $\lambda,\mu\ge 0$. 
Hence by (\ref{eq:tildeb}), $y=(\lambda-\mu)\tilde b+ z$ for some $z\in J^\perp$. If $y\in \tilde J\cap S_b$, we must have
 $1=\Tr yb=\lambda-\mu$, so that $y\in (\tilde b+J^\perp)\cap B(\Ha)^+=\tilde B$. The opposite inclusion is obvious. 

 (ii) It is clear that $B\subset \tilde{\tilde B}$ and $\tilde{\tilde B}$ is a base of $\tilde {\tilde Q}=\tilde{\tilde J}\cap B(\Ha)^+$. Consequently, it is enough to prove that $\tilde{\tilde J}=J$, since then $B$ and $\tilde {\tilde B}$ are two bases of the same cone. By the proof of (i), $\tilde J= \mathrm{span} \{\tilde b\}\vee J^\perp$ and similarly,
\[
\tilde{\tilde J}=\mathrm{span}\{b\}\vee \tilde J^\perp=\mathrm{span}\{b\}\vee (\{\tilde b\}^\perp\wedge J)=(\mathrm{span}\{b\}\vee \{\tilde b\}^\perp)\wedge J=J,
\]
here the third equality follows from $b\in J$ and the last equality  follows from the fact that $\{\tilde b\}^\perp$ is a subspace of codimension 1 not containing $b$,  so that $\mathrm{span}\{b\}\vee \{\tilde b\}^\perp=B_h(\Ha)$.

(iii) It is clear that $B\subseteq \bigcap_{\tilde b'\in ri(\tilde B)} S_{\tilde b'}$. 
If $a\in \bigcap_{\tilde b'\in ri(\tilde B)} S_{\tilde b'}$, then $a$ is a positive element such that 
$\Tr a\tilde b'=1$ for all
$\tilde b'\in cl(ri(\tilde B))=\tilde B$, hence $a\in \tilde{\tilde B}=B$. 

\end{proof}

We call $\tilde B$ the dual section of $B$.
Since $B$ is a base of $Q$, the base norm $\|\cdot\|_B$ is defined in $J$.
Next we show that this norm can be naturally extended to $B_h(\Ha)$.  For this, let us  define 
\begin{equation}\label{eq:Oa}
\Oe_B:=\{ x\in B_h(\Ha), x=x_1-x_2,\ x_1,x_2\in B(\Ha)^+, x_1+x_2\in B\}
\end{equation}

For $b\in B(\Ha)^+$, we define $\Oe_b:=\Oe_{\b}$.

\begin{lemma}\label{lemma:base}   We have
 \begin{enumerate}
\item [(i)] $\Oe_B=\{ x\in B_h(\Ha), \exists b'\in B,\ -b'\le x\le b'\}=\bigcup_{b'\in B}\Oe_{b'}$
\item [(ii)] The unit ball of the base norm $\|\cdot\|_B$ in $J$ is $\Oe_B\cap J$.
 \end{enumerate}

\end{lemma}

\begin{proof} 
(i) Let $x=x_1-x_2$ with $x_1+x_2=b'\in B$, then $-b'=-(x_1+x_2)\le x\le x_1+x_2=b'$. 
Conversely, let $-b'\le x\le b'$ for some $b'\in B$. Put  $x_\pm=1/2(b'\pm x)$, then $x_\pm\in B(\Ha)^+$, 
$x_+-x_-=x$ and $x_++x_-=b'\in B$.

(ii) By definiton, the unit ball of $\|\cdot\|_B$ in $J$  is the set of elements of 
$J$ of the form $x=\lambda b_1-(1-\lambda)b_2$, $b_1,b_2\in B$, $0\le \lambda\le 1$. Then clearly $x\in \Oe_B$, 
by putting
 $x_1=\lambda b_1$ and $x_2=(1-\lambda)b_2$. 
Conversely, let $x\in J$ be such that $-b'\le x\le b'$ for some $b'\in B$, and put $x_\pm=1/2(b'\pm x)$, 
then $x_\pm\in B(\Ha)^+\cap J=Q$. Let $x_\pm=\lambda_\pm b_\pm$, for $\lambda_\pm\ge 0$, $b_\pm\in B$,
then by applying the corresponding order unit $e_B$ to the equality $b'=x_++x_-$, we see that
 we must have $\lambda_++\lambda_-=1$, so that $\|x\|_B\le 1$.

\end{proof}

\begin{thm}\label{thm:norm_B_and_dual} Let $B$ be a faithful section and let $\tilde B$ be the dual section. Then $\Oe_B$ is the unit ball of a norm in $B_h(\Ha)$. The unit ball of the dual norm is $\Oe_{\tilde B}$.

\end{thm}

We will denote this norm by $\|\cdot\|_B$, note that Lemma \ref{lemma:base} (ii) justifies this notation.

\begin{proof}  It is clear that  $\Oe_B$ is convex and symmetric, that is, $-\Oe_B\subseteq \Oe_B$. 
Since $B$ is compact, $\Oe_B$  is  closed. If $x\in \Oe_B$, then  $x=x_1-x_2$ with
 $x_1,x_2\ge 0$, $x_1+x_2\in B$ and by (\ref{eq:basenorm_b}), 
 \[ 
\|x\|_{S_{\tilde b}}\le \|x_1\|_{S_{\tilde b}}+\|x_2\|_{S_{\tilde b}}=\Tr( x_1+x_2) \tilde b = 1
\]
for any $\tilde b\in ri(\tilde B)$, hence $\Oe_B$ is bounded. Moreover, since $b\in ri(B)$ is an order unit,
 for every $x\in B_h(\Ha)$ there is some $t>0$ such that $-tb\le x\le tb$, so that $\Oe_B$ is absorbing (see
 Lemma \ref{lemma:base} (i)). These facts imply that $\Oe_B$ is the unit ball of a norm.

To show duality of the norms $\|\cdot\|_B$ and $\|\cdot\|_{\tilde B}$, let $\Ha_2=\Ha\oplus \Ha$ and let 
$\Phi:B_h(\Ha_2)\to
B_h(\Ha)$ be the map defined by $\Phi(a\oplus b)= a+b$. 
Let $J_2=\Phi^{-1}(J)$, then $J_2$ is a subspace in $B_h(\Ha_2)$ and 
\[
J_2^\perp=\Phi^*(J^\perp)=\{ x\oplus x, \ x\in J^\perp\},
\]
 see \cite{ja}. Let 
 $\pi_2:B(\Ha_2)\to J_2^*=B(\Ha_2)|_{J_2^\perp}$ be the quotient map.

Let $\tilde b\in ri(\tilde B)$ and put $B_2=J_2\cap S_{\tilde b\oplus\tilde b}$. Then $B_2$  is a  base of $Q_2=J_2\cap B(\Ha_2)^+$ and it is clear that for $w_1,w_2\in B(\Ha)^+$, $w_1\oplus w_2\in B_2$ if and only if $w_1+w_2\in B$.
Let now $a\in B_h(\Ha)$, then $a\in \Oe_B^\circ$ 
if and only if $\Tr(a\oplus -a)w\le1$ for all  $w\in B_2$. Equivalently,
\[
\pi_2(a\oplus -a)\le_{Q_2^*} e_{B_2}=\pi_2(\tilde b\oplus\tilde b),
\]
 that is, there is  some $v\in J_2^\perp$ such that $a\oplus -a\le \tilde b\oplus\tilde b +v$. Since $v=x\oplus x$, $x\in J^\perp$,
 we obtain $\pm a\le \tilde b+x$. Note that we must have  $\tilde b+x\ge 0$: if $c$ is any element in $B(\Ha)^+$, then we have
 $\pm \Tr ca\le \Tr c(b+x)$, so that $\Tr c(b+x)$ cannot be negative. Hence $\pm a\le \tilde b+x\in \tilde B$, so that $a\in \Oe_{\tilde B}$, by Lemma \ref{lemma:base} (i).

\end{proof}

\begin{coro}\label{coro:norm_ext} Let $x\in B_h(\Ha)$. Then
\begin{enumerate}
\item[(i)] $\Oe_B=\bigcap_{\tilde b\in ri(\tilde B)} \Oe_{S_{\tilde b}}$,
\item[(ii)]$\|x\|_B=\sup_{\tilde b\in ri(\tilde B)}
\|x\|_{S_{\tilde b}}=\sup_{\tilde b\in \tilde B} \|\tilde b^{1/2}x\tilde b^{1/2}\|_1$
\item[(iii)] $\|x\|_B=\inf_{b\in ri(B)}\|x\|_b=\inf_{b\in B} \|b^{-1/2}xb^{-1/2}\|$.
\end{enumerate}
\end{coro}

\begin{proof} (i) It is easy to see from  Lemma \ref{lemma:base} that  
\begin{equation}\label{eq:ob}
\Oe_{ B}=\bigcup_{ b\in  B}\Oe_{ b}=cl(\bigcup_{ b\in ri( B)}\Oe_{ b}).
\end{equation}
Indeed, let $x\in B_h(\Ha)$ be such that $-b\le x\le b$ for some $b\in B$ and let $b'\in ri(B)$, then 
$b_\epsilon:=\epsilon b'+(1-\epsilon)b\in ri(B)$ for all $0<\epsilon<1$. Let $x'\in \Oe_{b'}$ be any element,
 then $x_\epsilon:= \epsilon x'+(1-\epsilon)x\in \Oe_{b_\epsilon}$ and $x=\lim_{\epsilon\to 0^+}x_\epsilon \in
cl(\bigcup_{b\in ri(B)}\Oe_b)$.

Since $A^\circ=(cl(conv(A)))^\circ$ for any subset $A\in B_h(\Ha)$ containing 0, we obtain by Theorem 
\ref{thm:norm_B_and_dual} that
\[
\Oe_B=\Oe_{\tilde B}^\circ= (\bigcup_{\tilde b\in ri(\tilde B)}\Oe_{\tilde b})^\circ=\bigcap_{\tilde b\in ri(\tilde B)} \Oe_{\tilde b}^\circ=\bigcap_{\tilde b\in ri(\tilde B)} \Oe_{S_{\tilde b}} 
\] 

(ii) Since $\Oe_B$ is the unit ball of $\|\cdot\|_B$, we get from (i) 
\begin{eqnarray*}
\|x\|_B&=&\inf\{\lambda>0, x\in \lambda \Oe_B\}=\inf\{\lambda>0, x\in \lambda \Oe_{S_{\tilde b}}, \forall \tilde b\in 
ri( \tilde B)\}\\
&=& \inf\{\lambda>0, \lambda\ge \|x\|_{S_{\tilde b}}, \forall \tilde b\in ri(\tilde B)\}=\sup_{\tilde b\in ri(\tilde B)}
\|x\|_{S_{\tilde b}}=\sup_{\tilde b\in \tilde B}\|\tilde b^{1/2}x\tilde b^{1/2}\|_1,
\end{eqnarray*}
the last equality follows from (\ref{eq:basenorm_b}) and continuity of the norm $\|\cdot\|_1$.

(iii) On the other hand, we get from Lemma \ref{lemma:base} and (\ref{eq:ob})
\begin{eqnarray*}
\|x\|_B&=&\inf\{\lambda>0, x\in \lambda \Oe_B\}=\inf\{\lambda>0, x\in \lambda \cup_{b\in ri(B)} \Oe_b\}\\
&=& \inf_{b\in ri(B)}\ \inf\{\lambda>0, x\in \lambda \Oe_b\}=\inf_{b\in ri(B)} \|x\|_b=\inf_{b\in B}\|b^{-1/2}xb^{-1/2}\|
\end{eqnarray*}
where the last equality follows by (\ref{eq:orderunitnorm_b}).

\end{proof}

\begin{coro} \label{coro:positive} For $a\in B(\Ha)^+$, we have
\[ 
\|a\|_B=\sup_{\tilde b\in \tilde B} \Tr  a\tilde b =\inf_{b\in B}2^{D_{max}(a\|b)}
\]
\end{coro}

\begin{proof}  We have
\[
\|a\|_B=\sup_{x\in \Oe_{\tilde B}} \Tr ax 
\]
Let $x\in \Oe_{\tilde B}$, then $x=x_1-x_2$, $x_1,x_2\in B(\Ha)^+$ and $x_1+x_2=:\tilde b_x\in \tilde B$, so that
\[
\Tr ax \le \Tr ax_1\le \Tr a\tilde b_x\le \sup_{\tilde b\in \tilde B}\Tr a\tilde b\le \sup_{y\in \Oe_{\tilde B}} \Tr ay=\|a\|_B
\]
Hence $\|a\|_B=\sup_{\tilde b\in \tilde B} \Tr a\tilde b$.
The second equality follows directly from Corollary \ref{coro:norm_ext} (iii) and the definition of $D_{max}$.

\end{proof}

We can also characterize the maximizer resp. minimizer in Corollary \ref{coro:positive}.

\begin{coro}\label{coro:minmaximizer} Let $a\in B(\Ha)^+$. 
\begin{enumerate}
\item[(i)] Let $\tilde b_0\in \tilde B$, then $\|a\|_B=\Tr a\tilde b_0$ if and only if there exists some $q\in Q$,
 such that $a\le q$ and $(q-a)\tilde b_0=0$. In this case, $q=\|a\|_Bb_0$ and  $\|a\|_B=
2^{D_{max}(a\|b_0)}$.
\item[(ii)] Let $b_0\in B$, then  $\|a\|_B=2^{D_{max}(a\|b_0)}$ if and only if there exists some $t>0$ and $\tilde b_0\in \tilde B$, such that
$a\le tb_0$ and $(tb_0-a)\tilde b_0=0$. In this case, $t=\|a\|_B=\Tr a\tilde b_0$.
\end{enumerate}

\end{coro}

\begin{proof} (i) Let  $\tilde b_0\in \tilde B$ be such that $\|a\|_B=\Tr a\tilde b_0$. Let $b_0\in B$  be such that
$\|a\|_B=2^{D_{max}(a\|b_0)}$, in particular, $a\le \|a\|_Bb_0$. Put $q=\|a\|_Bb_0$, then $q-a\ge 0$ and
$\Tr (q-a)\tilde b_0=0$. Since also $\tilde b_0\ge 0$, it follows that $(q-a)\tilde b_0=0$.

Conversely, suppose $q\in Q$ satisfies $a\le q$ and $(q-a)\tilde b_0=0$. Then $q=sb_0$ for some $b_0\in B$, $s\ge 0$.
Since $a\le sb_0$, we have 
\[
\|a\|_B\le s=\Tr a\tilde b_0\le \|a\|_B,
\]
so that $\Tr a\tilde b_0=\|a\|_B=s=2^{D_{max}(a\|b_0)}$.

(ii) is proved similarly.

\end{proof}

\section{Generalized channels}\label{sec:genchans}

Let $B$ be a  section of a base of $B(\Ha)^+$. A generalized channel with respect to $B$ (or a $B$-channel) is a completely positive map 
$\Phi:B(\Ha)\to B(\Ka)$ such that
$\Phi(B)\subseteq \Se(\Ka)$. Let $X_\Phi$ be the Choi matrix of $\Phi$, then  $\Phi$ is a generalized channel with respect to $B$ if and only if  $X_\Phi\ge 0$ and
\[
1=\Tr \Phi(b)=\Tr \Tr_\Ha[(I\otimes b^{\mathsf T})X_\Phi]=\Tr(I\otimes b^{\mathsf T})X_\Phi=\Tr b^{\mathsf T}\Tr_\Ka X_\Phi
\]
for all $b\in B$. Let
$\Ce_B(\Ha,\Ka)$ denote  the set of Choi matrices of all generalized channels with respect to $B$, then
\[
\Ce_B(\Ha,\Ka)=\{X\in B(\Ka\otimes \Ha)^+, \Tr_\Ka X\in \tilde B^{\mathsf T}\}.
\]
 Let us remark
that if $B$ is a section, then  $B^{\mathsf T}:=\{b^{\mathsf T}, b\in B\}$ is a section as well, here $b^{\mathsf T}$ denotes the transpose of $b$. Moreover, $\widetilde{B^{\mathsf T}}=\tilde B^{\mathsf T}$.
Note also that we have 
\begin{equation}\label{eq:B_tilde}
\Ce_B(\Ha,\mathbb C)=\tilde B^{\mathsf T},
\end{equation}
so that, in particular, $\Ce_B(\Ha,\mathbb C)$ is a section.

\begin{prop}\label{prop:genchan} Let $B$ be  a faithful section of a base of $B(\Ha)^+$. Then
$\Ce_B(\Ha,\Ka)$ is a faithful section of a base of $B(\Ka\otimes \Ha)^+$ and $\widetilde{\Ce_B(\Ha,\Ka)}=
\{I_\Ka\otimes b^{\mathsf T},\ b\in B\}$.

\end{prop}

\begin{proof}
It is easy to see that $I_\Ka\otimes B^{\mathsf T}=\{I_\Ka \otimes b^{\mathsf T}, b\in B\}$ is a faithful section of a base of $B(\Ka\otimes \Ha)^+$ and
\[
\Ce_B(\Ha,\Ka)=\{ X\in B(\Ka\otimes \Ha)^+, \Tr X(I\otimes b^{\mathsf T})=1, \forall b\in B\}=\widetilde{I_\Ka\otimes B^{\mathsf T}}.
\]
 The proof now follows by Lemma \ref{lemma:B_tilde_B} (i) and 
(ii).

\end{proof}

Let now $X\in B_h(\Ka\otimes \Ha)$ and let $\Phi:B(\Ha)\to B(\Ka)$ be the corresponding Hermitian map. By Corollary 
\ref{coro:norm_ext} and Proposition \ref{prop:genchan},
\[
\|X\|_{\Ce_B(\Ha,\Ka)} = \sup_{b\in B } \|(I\otimes (b^{\mathsf T})^{1/2})X(I\otimes (b^{\mathsf T})^{1/2})\|_1
\]
and we have
\[
(I\otimes (b^{\mathsf T})^{1/2})X(I\otimes (b^{\mathsf T})^{1/2})=(\Phi\otimes id_\Ha)(\sigma_b),
\]
where $\sigma_b=|\psi_b\>\<\psi_b|$, with 
\[
|\psi_b\>=\sum_i |i\>\otimes (b^{\mathsf T})^{1/2}|i\>=\sum_i b^{1/2} |i\>\otimes |i\>\in \Ha\otimes \Ha.
\]
Hence $\sigma_b\in B(\Ha\otimes\Ha)^+$ and $\Tr_1\sigma_b=b^{\mathsf T}\in B^{\mathsf T}$, so that $\sigma_b\in 
\Ce_{\tilde B}(\Ha,\Ha)$. Conversely, if $\sigma=|\varphi\>\<\varphi|\in \Ce_{\tilde B}(\Ha,\Le)$ for some Hilbert space $\Le$, 
then there is some 
linear map $R:\Ha\to \Le$ satisfying $R^*R=b\in B$ and such that $|\varphi\>=\sum_i R|i\>\otimes|i\>$. Let $U:\Ha\to\Le$
 be an isometry such that $R=Ub^{1/2}$, then 
\[
|\varphi\>=\sum_i R|i\>\otimes |i\>=\sum_i Ub^{1/2}|i\>\otimes |i\>=(U\otimes I)|\psi_b\>
\]

\begin{thm}\label{thm:norm_genchan} Let $X\in B_h(\Ka\otimes\Ha)$ and let $\Phi$ be the corresponding Hermitian map 
$B(\Ha)\to B(\Ka)$. Let $\Le$ be any Hilbert space with  $\dim(\Le)=\dim(\Ha)$.
Then
\begin{align*}
\|X\|_{\Ce_B(\Ha,\Ka)}&=\sup_{\dim(\Le')<\infty} \
\sup_{\sigma\in \Ce_{\tilde B}(\Ha,\Le')} \|(\Phi\otimes id_{\Le'})(\sigma)\|_1 \\
&=   \sup_{\sigma\in \Ce_{\tilde B}(\Ha,\Le)}\|(\Phi\otimes id_{\Le})(\sigma)\|_1\\
\end{align*}
and the dual norm is $\|X\|_{\Ce_B(\Ha,\Ka)}^*=\|X\|_{I_\Ka\otimes B^{\mathsf T}}$.
Moreover, if $X\ge 0$ then
\[
\|X\|_{\Ce_B(\Ha,\Ka)}= \sup_{b\in B} \Tr \Phi(b)= \inf_{Y\in \Ce_B(\Ha,\Ka)} 2^{D_{max}(X\|Y)}
\]
and
\[
\|X\|_{I\otimes B^{\mathsf T}}=\inf_{b\in B} 2^{D_{max}(X\| I\otimes b^{\mathsf T})}
= \sup_{Y\in \Ce_B(\Ha,\Ka)}\Tr XY=\sup_{S}\<\psi| X_{S^*\circ \Phi}|\psi\>,
\]
where the last supremum is taken over the set of all $B$-channels $B(\Ha)\to B(\Ka)$.

\end{thm}

{\it Proof.} From what was said above, it is easy to see that
\[
\|X\|_{\Ce_B(\Ha,\Ka)}=  \sup_{|\varphi\>\<\varphi|\in \Ce_{\tilde B}(\Ha,\Le)}\|(\Phi\otimes id_{\Le})(|\varphi\>\<\varphi|)\|_1
\]
with $\dim(\Le)=\dim(\Ha)$. We will show that
\[
\sup_{|\varphi\>\<\varphi|\in \Ce_{\tilde B}(\Ha,\Le')}\|(\Phi\otimes id_{\Le'})(|\varphi\>\<\varphi|)\|_1\le
\sup_{|\varphi\>\<\varphi|\in \Ce_{\tilde B}(\Ha,\Le)}\|(\Phi\otimes id_{\Le})(|\varphi\>\<\varphi|)\|_1
\]
whenever $\dim(\Le')\ge \dim(\Le)$. The proof is almost the same as the proof of \cite[Theorem 5]{watrous}, we 
include it here for completeness.

So let $\dim(\Le')\ge \dim(\Le)=\dim(\Ha)$, then there is some $\varphi_0\in  \Ha\otimes\Le'$,
with $|\varphi_0\>\<\varphi_0|\in \Ce_{\tilde B}(\Ha,\Le')$
 such that 
\[
\sup_{|\varphi\>\<\varphi|\in \Ce_{\tilde B}(\Ha,\Le')} \|\Phi\otimes id_{\Le'}(|\varphi\>\<\varphi|)\|_1=
\|\Phi\otimes id_{\Le'}(|\varphi_0\>\<\varphi_0|)\|_1
\]
Let $|\varphi_0\>=\sum_{i=1}^m s_i|\varphi_i\>\otimes |\xi_i\>$ be the Schmidt decomposition of $\varphi_0$, with 
$\{|\varphi_i\>\}$ and $\{|\xi_i\>\}$ orthonormal sets in $\Ha$ resp. $\Le'$ and $m=\dim(\Ha)$. Then 
$|\varphi_0\>\<\varphi_0|=\sum_{i,j} |\varphi_i\>\<\varphi_j|\otimes|\xi_i\>\<\xi_j|$ and 
\[
(\Tr_{\Le'}|\varphi_0\>\<\varphi_0|)^{\mathsf T}=
(\sum_i s_i|\varphi_i\>\<\varphi_i|)^{\mathsf T}\in \tilde {\tilde B}=B.
\]

Let $\{|e_i\>, i=1,\dots,m\}$ be an ONB in $\Le$. Define the linear map $U: \Le'\to \Le$ by
$U=\sum_{i=1}^m|e_i\>\<\xi_i|$, then $U^*U=\sum_i|\xi_i\>\<\xi_i|$ is the projection in $\Le'$ onto the subspace spanned
 by the vectors $|\xi_i\>$, $i=1,\dots,m$ and $(I\otimes U^*U)|\varphi_0\>=|\varphi_0\>$. Put 
$\varphi_U:=(I\otimes U)|\varphi_0\>=\sum_i |\varphi_i\>\otimes |e_i\>$, then it is easy to see that $|\varphi_U\>\<\varphi_U|\in
 \Ce_{\tilde B}(\Ha,\Le)$. Now we have
\begin{eqnarray*}
\sup_{|\varphi\>\<\varphi|\in \Ce_{\tilde B}(\Ha,\Le)} \|\Phi\otimes id_\Le(|\varphi\>\<\varphi|)\|_1&\ge& 
\|\Phi\otimes id_\Le(|\varphi_U\>\<\varphi_U|)\|_1\\
&\ge& \|(I\otimes U^*)(\Phi\otimes id_\Le)(|\varphi_U\>\<\varphi_U|)(I\otimes U)\|_1\\
&=& \|\Phi\otimes id_{\Le'}((I\otimes U^*)|\varphi_U\>\<\varphi_U|(I\otimes U))\|_1 \\
&=&\|\Phi\otimes id_{\Le'}(|\varphi_0\>\<\varphi_0|)\|_1\\
&=&\sup_{|\varphi\>\<\varphi|\in \Ce_{\tilde B}(\Ha,\Le')} \|\Phi\otimes id_{\Le'}(|\varphi\>\<\varphi|)\|_1
\end{eqnarray*}

Next, let $Y$ be any element in $\Ce_{\tilde B}(\Ha,\Le')$, then the corresponding map $\xi:B(\Ha)\to B(\Le')$
has the form
\[
\xi(a)=\sum_{i=1}^N V_iaV_i^*, \qquad a\in B(\Ha)
\]
where $V_i: \Ha\to\Le'$ are linear maps such that $\sum_i V_i^*V_i\in B$.
Let $\Le_0'$ be a Hilbert space with $\dim(\Le'_0)=N$ and let $\{|f_j\>, j=1,\dots,N\}$ be an ONB
 in $\Le'_0$. Define $V=\sum_{j=1}^N V_j\otimes |f_j\>$, then $V$ is a linear map
$\Ha\to \Le'\otimes \Le'_0$ with $V^*V=\sum_iV_i^*V_i\in B$. Let  $\mathcal V(a)=V aV^*$ and  let $Z$ be the Choi matrix 
of $\mathcal V$, then $Z$ is a rank one element in $\Ce_{\tilde B}(\Ha,\Le'\otimes
\Le_0')$. Moreover,
$\xi(a)=\Tr_{\Le_0'} V a V^*$ and $Y=\Tr_{\Le_0'} Z$. It follows that
\begin{eqnarray*}
\|(\Phi\otimes id_{\Le'})(Y)\|_1&=&\|(\Phi\otimes id_{\Le'})(\Tr_{\Le_0'} Z)\|_1=
\|\Tr_{\Le_0'}(\Phi\otimes id_{\Le'\otimes \Le_0'})(Z)\|_1\\
&\le& \|(\Phi\otimes id_{\Le'\otimes \Le_0'})(Z)\|_1\le \|X\|_{\Ce_B(\Ha,\Ka)}
\end{eqnarray*}
We now have
\begin{eqnarray*}
\|X\|_{\Ce_B(\Ha,\Ka)}&=&\sup_{|\varphi\>\<\varphi|\in \Ce_{\tilde B}(\Ha,\Le)}\|(\Phi\otimes id_\Le)(|\varphi\>\<\varphi|)\|_1
\le \sup_{\sigma\in \Ce_{\tilde B}(\Ha,\Le)}\|(\Phi\otimes id_\Le)(\sigma)\|_1\\
&\le& \sup_{\dim(\Le')<\infty}\ \sup_{\sigma\in \Ce_{\tilde B}(\Ha,\Le')}\|(\Phi\otimes id_{\Le'})(\sigma)\|_1\le 
\|X\|_{\Ce_B(\Ha,\Ka)}
\end{eqnarray*}

The expression for the dual norm follows by Proposition \ref{prop:genchan}. 
Suppose now that $X\ge0$, then by Corollary \ref{coro:positive}
\begin{eqnarray*}
\|X\|_{\Ce_B(\Ha,\Ka)}&=&\sup_{b\in B}\Tr X(I\otimes b^{\mathsf T})=\inf_{Y\in \Ce_B(\Ha,\Ka)}2^{ D_{max}(X\| Y)}\\
\|X\|_{I\otimes B^{\mathsf T}}&=&\sup_{Y\in \Ce_B(\Ha,\Ka)}\Tr XY=\inf_{b\in B}2^{D_{max}(X\| I\otimes b^{\mathsf T})}
\end{eqnarray*}
By (\ref{eq:choirep}), $\Tr X(I\otimes b^{\mathsf T})=\Tr \Tr_\Ha X(I\otimes b^{\mathsf T})=\Tr \Phi(b)$. Moreover, let
 $Y\in \Ce_B(\Ha,\Ka)$ and let $S$ be the corresponding $B$-channel, then
\[
\Tr XY=\Tr X(S\otimes id)(\Psi)=\Tr (S^*\otimes id)(X)\Psi=\<\psi, X_{S^*\circ \Phi}\psi\>.
\]

\qed

\subsection{Channels}

Let $B=\Se(\Ha)$, then generalized channels are the usual channels. 
In this case, we denote $\Ce_B(\Ha,\Ka)$ by $\Ce(\Ha,\Ka)$. Note that $\tilde B=\{I\}$ and
$\Ce_{\tilde B}(\Ha,\Ka)=\Se(\Ka\otimes \Ha)$.

By Proposition \ref{prop:genchan},   $\Ce(\Ha,\Ka)$ is a faithful section of a base of $B(\Ka\otimes \Ha)^+$ and 
\[
\widetilde {\Ce(\Ha,\Ka)}=\{I_\Ka\otimes \rho, \rho\in \Se(\Ha)\}.
\]
Furthemore, let $X\in B_h(\Ka\otimes \Ha)$ and let $\Phi:B(\Ha)\to B(\Ka)$ be the corresponding Hermitian map. Then
by Theorem \ref{thm:norm_genchan}, 
\[
\|X\|_{\Ce(\Ha,\Ka)} = 
\sup_{\sigma\in \Se(\Ha\otimes \Le)} \|(\Phi\otimes id_{\Le})(\sigma)\|_1=\|\Phi\|_{\diamond}
\]
with $\dim(\Le)=\dim(\Ha)$.
For the dual norm, we have
\[
\|X\|_{I\otimes \Se(\Ha)}=\inf_{\rho\in \Se(\Ha)}\inf\{\lambda>0,\ -\lambda(I\otimes \rho)\le X\le \lambda(I\otimes \rho)\} 
\]
If $\sigma\in B(\Ka\otimes \Ha)^+$, we obtain
\[
\|\sigma\|_{I\otimes \Se(\Ha)}=
\inf_{\rho\in \Se(\Ha)}2^{D_{max}(\sigma\| I\otimes \rho)}= 2^{-H_{min}(\Ka|\Ha)_\sigma}
\]
 where
$ H_{min}(\Ka|\Ha)_\sigma$ is the conditional min-entropy, see \cite{renner}. 

\subsection{Quantum supermaps} 
Let $\Ha_0,\Ha_1,\dots$ be a sequence of finite dimensional Hilbert spaces. For each $n\ge 1$, we define the sets
 $\Ce(\Ha_0,\dots,\Ha_n)$ as follows: $\Ce(\Ha_0,\Ha_1)$ is, as before,  the set of Choi matrices of channels
$B(\Ha_0)\to B(\Ha_1)$. For $n>1$, we define $\Ce(\Ha_0,\dots,\Ha_n)$ as 
the set of Choi matrices of cp maps $B(\Ha_{n-1}\otimes\dots\otimes\Ha_0)\to B(\Ha_n)$ that map
$\Ce(\Ha_0,\dots,\Ha_{n-1})$ into $\Se(\Ha_n)$. Such maps were called quantum supermaps in \cite{ja}\footnote{Note that this definition is slightly different from the notion of  supermap introduced in \cite{daria_super}}  and it was proved that
 for $n=2N-1$ we get precisely the set of deterministic quantum $N$-combs  for the sequence
$\{\Ha_0,\dots,\Ha_{2N-1}\}$. If $n=2N$, we get the set of $N+1$-combs for
$\{\mathbb C,\Ha_0,\dots,\Ha_{2N}\}$.

Let us fix the sequence $\Ha_0,\Ha_1,\dots$ and for this, put $\Ce_n=\Ce(\Ha_0,\dots,\Ha_n)$.
By using repeatedly Proposition \ref{prop:genchan}, we see that
 $\Ce_n$ is a faithful section of a  base of $B(\Ha_{n}\otimes\dots\otimes \Ha_0)^+$ and
 \[
\Ce_{n+1}=\Ce_{\Ce_n}(\Ha_{n}\otimes\dots\otimes\Ha_0,\Ha_{n+1}).
\]
 Moreover, by Proposition \ref{prop:genchan}, 
\[
\widetilde{\Ce_n}=I_{\Ha_n}\otimes \Ce_{n-1}=\Ce(\Ha_0,\dots,\Ha_n,\mathbb C)
\]
(note that $\Ce_{n-1}^{\mathsf T}=\Ce_{n-1}$, the last equality above follows from
 (\ref{eq:B_tilde})). For $n=2N-1$, 
this corresponds to the set of $N$-round nonmeasuring co-strategies of \cite{guwat, gutoski}.
Note also that for any finite dimensional Hilbert space $\Le'$,
\begin{eqnarray*}
\Ce_{\tilde \Ce_n}(\Ha_n\otimes\dots\otimes\Ha_0,\Le')&=&\{ Y\ge 0, \Tr_{\Le'} Y\in \Ce_n=\Ce_{\Ce_{n-1}}(\Ha_{n-1}\otimes\dots\otimes\Ha_0,\Ha_n)\}\\
&=&\{ Y\ge 0, \Tr_{\Ha_n}(\Tr_{\Le'} Y)\in \widetilde{\Ce_{n-1}} \}\\
&=&\Ce(\Ha_0,\dots,\Ha_n\otimes \Le')
\end{eqnarray*}

 Now we obtain  the following expressions for the corresponding
 norm and its dual.

\begin{thm}\label{thm:diamond_N} Let $n\ge 2$. Let $X\in B_h(\Ha_n\otimes\dots\otimes \Ha_0)$ and let 
$\Phi: B(\Ha_{n-1}\otimes\dots\otimes \Ha_0)\to B(\Ha_n)$ be the corresponding map. We have
\begin{eqnarray*}
\|X\|_{\Ce(\Ha_0,\dots,\Ha_n)}&=& \sup_{Y_1,Y_2\ge 0, Y_1+Y_2\in  \Ce(\Ha_0,\dots,\Ha_n,\mathbb C)}\Tr X(Y_1-Y_2)\\ 
&=& \sup_{Y\in \Ce(\Ha_0,\dots,\Ha_n,\mathbb C)} \|Y^{1/2}XY^{1/2}\|_1\\
&=& \inf_{Y\in \Ce(\Ha_0,\dots,\Ha_n)}\inf\{\lambda>0, -\lambda Y\le X\le \lambda Y\}\\
&=&\sup_{\dim(\Le')<\infty}\ \sup_{Y\in \Ce(\Ha_0,\dots,\Ha_{n-2},\Ha_{n-1}\otimes \Le')}
\|(\Phi\otimes id_{\Le'})(Y)\|_1\\
&=&\sup_{Y\in \Ce(\Ha_0,\dots,\Ha_{n-2},\Ha_{n-1}\otimes \Le)}
\|(\Phi\otimes id_{\Le})(Y)\|_1
\end{eqnarray*}
where $\dim(\Le)=\dim(\Ha_{n-1}\otimes\dots\otimes\Ha_0)$. Moreover, the dual norm is
\[
\|X\|_{I_{\Ha_n}\otimes \Ce(\Ha_0,\dots,\Ha_{n-1})}=\|X\|_{\Ce(\Ha_0,\dots,\Ha_n,\mathbb C)}
\]

\end{thm}

{\it Proof.}
Duality of the norms is obtained from Theorem \ref{thm:norm_genchan}, this also implies the first equality.
 Next two equalities follow by Corollary \ref{coro:norm_ext}. The rest follows by Theorem \ref{thm:norm_genchan}.

\qed

For $n=2N-1$, first two expressions are exactly  the $N\diamond$-norm as obtained in  \cite{gutoski} and \cite{daria_testers}. Duality of the norms corresponding to strategies and co-strategies was 
 also obtained in \cite{gutoski}.

\section{A general quantum decision theory}

As before, let $B$ be a faithful section of a base of $B(\Ha)^+$, $J=\mathrm{span}(B)$ and $Q=J\cap B(\Ha)^+$.
 As we have seen, elements of $B$ may represent certain quantum devices and
it is therefore reasonable to consider the following definitions.

Let $\{b_\theta, \theta\in \Theta\}\subset B$ be a parametrized family, for simplicity, we will suppose that the set of
 parameters $\Theta$ is finite. If $B$ is the set of states, the couple 
$\ee=(\Ha,\{b_\theta, \theta\in \Theta\})$ is called an experiment and is
 interpreted as an a priori information on the true state of the system. Accordingly, for a section  $B$, we will define a  generalized experiment as a triple $\ee=(\Ha, B,\{ b_\theta,\theta\in \Theta\})$.

Another ingredient of decision theory is a (finite) set $D$, the set of possible decisions.
A decision procedure $\mathbf m$ is a procedure by which we pick some decision $d\in D$,
 with  probability based on the ''true value'' of $b$. That is, $\mathbf m$ is a map $B\to \Pe(D)$, 
where $\Pe(D)$ is the set of probability measures on $D$, such a map will be called a measurement on $B$, 
with values in $D$. The payoff obtained if $d\in D$ is chosen while the true value is $\theta\in \Theta$ is given by
the payoff function  $w: \Theta\times D\to[0,1]$,  the pair $(D,w)$ is called
 a (classical) decision problem. Let $\lambda$ be an a priori probability distribution 
on $\Theta$. The task is to maximize  the average payoff, that is the value of
\begin{equation}\label{eq:average_gain}
\mathcal L_{\ee,\lambda,w}(\mathbf m):=\sum_{\theta,d} \lambda_\theta w(\theta,d)\mathbf m(b_\theta)_d
  \end{equation}
over all measurements $\mathbf m: B\to \Pe(D)$.

It is quite clear that any measurement $\mathbf m$ on $B$ is given by a collection $\{\mathbf m_d, d\in D\}$ 
of elements in $Q^*$ such that $\mathbf m(b)_d=\<\mathbf m_d,b\>$ and that we must have
$\sum_d\mathbf m_d=e_B$. Similarly as it was shown in \cite{ja}, any measurement is given by a collection 
 $\{M_d,d\in D\}\subset B(\Ha)^+$ such that $\mathbf m_d=\pi(M_d)$ and $\pi(\sum_d M_d)=e_B$, that is
\[
\sum_dM_d\in \pi^{-1}(e_B)\cap B(\Ha)^+=\tilde B.
\]
 Any such collection of positive operators will be called 
a generalized POVM
   (with respect to $B$), or a $B$-POVM. It is also clear that any $B$-POVM defines a measurement on $B$ (but  it may 
happen that different generalized POVMs define the same measurement, see \cite{ja}). If $B=\Se(\Ha)$, we obtain
 a (usual) positive operator valued measure (POVM) $M=\{ M_d, d\in D\}\subset B(\Ha)^+$, $\sum_d M_d=I$.

Let us denote by $\Me_B(\Ha,D)$ the set of all generalized POVMs with respect to $B$ with values in $D$ and let 
$\{M_d, d\in D\}\in \Me_B(\Ha,D)$. Let us denote
\begin{equation}\label{eq:gPOVM}
M=\sum_{d\in D} |d\>\<d|\otimes M_d^{\mathsf T}\in B(\Ha_D\otimes \Ha)^+,
\end{equation}
 where $\Ha_D$ is a Hilbert space with $\dim(\Ha_D)=|D|$ and $\{|d\>, d\in D\}$ an ONB in $\Ha_D$. 
Then it is clear that $M$ is a block-diagonal element in $\Ce_B(\Ha,\Ha_D)$. Conversely, it is clear that
if $X=\sum_d |d\>\<d|\otimes X_d\in \Ce_B(\Ha,\Ha_D)$, then $\{X_d^{\mathsf T}, d\in D\}\in \Me_B(\Ha,D)$.
In this way, we identify $\Me_B(\Ha,D)$ with the subset of block-diagonal elements in $\Ce_B(\Ha,\Ha_D)$.

Let now $(D,w)$ be a decision problem and let $\mathbf m$ be a decision procedure  with  corresponding
$B$-POVM $M$. Then the average payoff is computed as
\[
\mathcal L_{\ee,\lambda,w}(\mathbf m)=\Le_{\ee,\lambda,w}(M):=\sum_{\theta,d} \lambda_\theta w(\theta,d)\Tr M_db_\theta=
\Tr \xi_{\ee,\lambda,w}M^{\mathsf T}
\]
where 
\[
\xi_{\ee,\lambda,w}=\sum_\theta\sum_d  \lambda_\theta w(\theta,d)|d\>\<d|\otimes b_\theta=
\sum_d|d\>\<d|\otimes\sum_\theta\lambda_\theta w(\theta,d)b_\theta\in B(\Ha_D\otimes \Ha)^+.
\]

More generally, let $\mathcal D$ be a Hilbert space, $\dim(\De)=k$ and let $W$ be a function $W: \theta \mapsto W_\theta\in B(\De)^+$, with $W_\theta\le I$.  We call the pair $(\De, W)$ a quantum decision problem, \cite{matsumoto}.
 A decision procedure is now a $B$-channel $\Phi: B(\Ha)\to B(\De)$ and the average payoff of $\Phi$ is  given by
\[
\mathcal L_{\ee,\lambda,W}(\Phi)= \sum_\theta\lambda_\theta \Tr\Phi(b_\theta)W_\theta
\]
If $X\in \Ce_B(\Ha,\De)$ is the Choi matrix of $\Phi$, then the average payoff has the form
\begin{align}
\mathcal L_{\ee,\lambda,W}(\Phi)=\mathcal L_{\ee,\lambda,W}(X)&:=  \sum_\theta\lambda_\theta \Tr(W_\theta\Tr_\Ha[(I_\De\otimes b_\theta^{\mathsf T})X])\notag \\\label{eq:maxpayoff}
&=
\sum_\theta\Tr(\lambda_\theta W_\theta\otimes b_\theta^{\mathsf T})X=\Tr \xi_{\ee,\lambda,W}X^{\mathsf T},
\end{align}
where 
\[
\xi_{\ee,\lambda,W}=\sum_\theta \lambda_\theta W_\theta^{\mathsf T}\otimes b_\theta\in B(\De\otimes \Ha)^+.
\]

It is easy to see that the set of quantum decision problems contains
also classical ones: Let $(D,w)$ be a classical decision problem and let $\Ha_D$ be as before. 
Let $W_\theta:=\sum_{d\in D} w(\theta,d) |d\>\<d|$, then $(\Ha_D,W)$ is a quantum decision problem and 
$\xi_{\ee,\lambda,W}=\xi_{\ee,\lambda,w}$. Let $X\in \Ce_B(\Ha,\Ha_D)$ and   $X=\sum_{c,d\in D} |c\>\<d|\otimes X_{cd}$ $X_{cd}\in B(\Ha)$. Since $\xi_{\ee,\lambda,w}$ is block-diagonal, we have 
\[
 \mathcal L_{\ee,\lambda,W}(X)=\mathcal L_{\ee,\lambda,w}(M),
 \]
where $M=\sum_d|d\>\<d|\otimes X_{dd}$ is a $B$-POVM. In other words, for a classical decision
problem one cannot get better results by considering quantum decision procedures. Conversely, let $(\mathcal D, W)$ be a 
quantum decision problem such that all the operators $W_\theta$ commute. Then there is a basis of $\mathcal D$ 
with respect to which all the operators $W_\theta$ are given by diagonal matrices, and the problem is equivalent 
to a classical problem, in the sense that the we obtain   the same average payoffs. Hence we can view the set of classical
decision problems as the subset of quantum decision problems such that the payoff function $W$ has commutative range.

\begin{thm}\label{thm:payoff}  Let $\ee=(\Ha,B,\{b_\theta,\theta\in \Theta\})$ be a generalized experiment and let 
$(\De,W)$ be a quantum decision problem. Then  the maximal average payoff is given by 
\[
\Le_{\ee,\lambda,W}:=\max_{X\in \Ce_B(\Ha,\De)} \mathcal L_{\ee,\lambda,W}(X)=\|\xi_{\ee,\lambda,W}\|_{I_\De\otimes B}
\]
 If $(\De,W)$ is classical, then
\[
\mathcal L_{\ee,\lambda,W}=\inf_{b\in B} \sup_{d\in D} 
2^{D_{\max}(\sum_\theta \lambda_\theta w(\theta,d)
b_\theta\| b)}
\]
\end{thm}

{\it Proof.} By (\ref{eq:maxpayoff}), the  maximal average payoff is given by
\[
\mathcal L_{\ee,\lambda,W}=\max_{X\in \Ce_B(\Ha,\De)}\Tr \xi_{\ee,\lambda,W} X^{\mathsf T}
=\|\xi_{\ee,\lambda,W}\|_{I_\De\otimes B},
\]
 the last equality follows by Corollary \ref{coro:positive} and Proposition \ref{prop:genchan}. 
If $(\De,W)$ is classical, then we may suppose that the matrices $W_\theta$ are diagonal. Then
 $\xi_{\ee,\lambda,W}=\sum_d |d\>\<d| \otimes \xi_{\ee,\lambda,W,d}$ is block-diagonal, where 
$\xi_{\ee,\lambda,W,d}=\sum_\theta \lambda_\theta w(\theta,d)b_\theta$. By 
 Corollary \ref{coro:positive}, and definition of $D_{max}$, 
\begin{align*}
\|\xi_{\ee,\lambda,W}\|_{I_\De\otimes B}&=\inf_{b\in B} 2^{D_{max}(\xi_{\ee,\lambda,W}\| I_\De\otimes b)}=
\inf_{b\in B}\inf\{\gamma>0, \xi_{\ee,\lambda,W,d}\le 2^\gamma b, \forall d\in D\}\notag\\
&= \inf_{b\in B}\sup_{d\in D} 2^{D_{max}(\xi_{\ee,\lambda,W,d}\| b)}
\end{align*}

\qed
We can also use Corollary \ref{coro:minmaximizer} to characterize  decision procedures that maximize  average
 payoff, we will call such procedures optimal with respect to $(\ee,\lambda, W)$.

\begin{coro}\label{coro:optimal}
Let $(\De,W)$ be a decision problem and let $X\in \Ce_B(\Ha,\De)$. Then $X$ is optimal with respect to $(\ee,\lambda,W)$ if and only if
there is some element $q\in Q$ such that $\xi_{\ee,\lambda,W}\le I_\De\otimes q$ and 
\begin{equation}\label{eq:optimal_proc}
((I\otimes q)-\xi_{\ee,\lambda,W})X^{\mathsf T}=0.
\end{equation}
If $(\De,W)$ is classical, then a $B$-POVM $(M_1,\dots,M_{\dim(\De)})$
 is optimal if and only if there is some $q\in Q$ such that $\sum_\theta\lambda_\theta w(\theta,d)b_\theta\le q$ for all $d$ and
\begin{equation}\label{eq:optimal_proc_class}
q\sum_d M_d= \sum_\theta\lambda_\theta b_\theta \sum_d w(\theta,d)M_d.
\end{equation}

\end{coro}

\begin{proof} The first part follows directly by Theorem \ref{thm:payoff} and Corollary \ref{coro:minmaximizer}. If $(\De,W)$ is
 classical, then $\xi_{\ee,\lambda,W}$ is block-diagonal, so that $\xi_{\ee,\lambda,W}\le I\otimes q$ if and only if
 each block is majorized by $q$, that is $\sum_\theta\lambda_\theta w(\theta,d)b_\theta\le q$. Moreover, (\ref{eq:optimal_proc_class}) implies that 
 \[
 \sum_d \Tr(q-\sum_\theta\lambda_\theta \rho_\theta w(\theta,d))M_d=0.
 \]
Since this is a sum of nonnegative elements, it is zero if and only if each summand is equal to zero. Again by positivity, 
this is equivalent to (\ref{eq:optimal_proc}).

\end{proof}

In particular, in the case $B=\Se(\Ha)$, we obtain the following optimality condition for POVMs.

\begin{coro}\label{coro:povms} Let $\ee=\{\sigma_\theta, \theta\in \Theta\}$ be an experiment and let $(D,w)$ be a 
 classical decision problem. Then a POVM $\{M_d, d\in D\}$ is optimal with respect to $(\ee,\lambda,W)$ if and only if
$q:=\sum_d\sum_\theta\lambda_\theta\sigma_\theta w(\theta,d)M_d$ is hermitian and such that $\sum_\theta\lambda_\theta \sigma_\theta w(\theta,d)\le q$ for all $d$.

\end{coro}

\begin{rem}\label{rem:payoff_loss} Sometimes the function $W$ is interpreted as loss rather than payoff, then $\Le_{\ee,\lambda,W}(\Phi)$
 is the average loss of the procedure $\Phi$ which has to be minimized. Let $W'_\theta=I_\De-W_\theta$, then $\theta\mapsto
 W_\theta$ is again a payoff (or loss) function and we have
\begin{align*}
\min_\Phi \Le_{\ee,\lambda,W}=&\min_{\Phi} \sum_\theta\lambda_\theta \Tr\Phi(b_\theta)W_\theta
=
\min_{\Phi}\sum_\theta\lambda_\theta \Tr\Phi(b_\theta)(I-W'_\theta)\\
&=1-\max_\Phi \mathcal L_{\ee,\lambda,W'}(\Phi)= 1-\|\xi_{\ee,\lambda,W'}\|_{I_\De\otimes B}
\end{align*}
Moreover, an optimal procedure $\Phi$ that minimizes the loss is a maximizer for $\Le_{\ee,\lambda,W'}$, hence satisfies the 
 conditions of Corollary \ref{coro:optimal}, with $W$ replaced by $W'$. Note that then the codtition from Corollary 
\ref{coro:povms} is the same as obtained in \cite{holevo1}.
\end{rem}

Let $\{M_d, d\in D\}$ be a $B$-POVM  with  $\sum_d M_d=c\in \tilde B$. Then since $0\le M_d\le c$ for all $d$, 
 we have 
\[
M_d=c^{1/2} \Lambda_d c^{1/2},\qquad d\in D,
\]
 where $\Lambda_d:=c^{-1/2}M_d c^{-1/2}$ defines a (usual) POVM on the support
 $\supp c$ of $c$. It follows that  $\Tr xM_d=\Tr c^{1/2}xc^{1/2}\Lambda_d$, that is, we can decompose the measurement defined
 by $\{M_d\}$ into a cp map $\chi_c: x\mapsto c^{1/2}xc^{1/2}$ followed by the usual measurement given by $\{\Lambda_d\}$,
 note that $\chi_c\in \Ce_B(\Ha,\supp c)$ so that $\chi_c$ maps a generalized experiment $\ee=(\Ha,B,\{b_\theta,\theta\in \Theta\})$ onto an 
ordinary experiment $\ee_c:=\{\supp c, \Se(\supp c),\{\chi_c(b_\theta),\theta\in \Theta\})$. We write this decomposition as $M=\Lambda \circ \chi_c$.  Such a decoposition was also  used in \cite{daria_testers} in the case of testers and in 
 \cite{ja} for  generalized POVMs.
Using this decomposition, we obtain the following optimality condition for $B$-POVMs.

\begin{coro}\label{coro:opt_class_dec}
Let $(D,w)$ be a classical decision problem and let $M\in \Me_B(\Ha,D)$ with decomposition $M=\Lambda\circ\chi_c$. 
Suppose $c$ is invertible  and let 
$\ee_c:=(\Ha,\{\sigma_\theta:=\chi_c(b_\theta),\theta\in \Theta\}) $.
Then $M$ is optimal for $(\ee,\lambda,w)$ if and only if $\Lambda$ is optimal for $(\ee_c,\lambda,w)$ and\[
\sum_d\sum_\theta \lambda_\theta w(\theta,d) \sigma_\theta \Lambda_d \in \chi_c(J)
\] 
\end{coro}

\begin{proof} Directly by Corollary \ref{coro:optimal} and \ref{coro:povms}.
 
\end{proof}

\begin{ex}[Multiple hypothesis testing]

  Suppose a family  $\{b_1,\dots,b_k\}$ of elements in $B$ is given and the task is to decide which is the true one,
 moreover, given
 some $\lambda\in \Pe(\{1,\dots k\})$, we want to minimize the average probability
 of making an error.  
In this case, put $\ee=(\Ha,B,\{b_1,\dots,b_k\})$,
$\Theta=D=\{1,\dots,k\}$ and the loss function is $w(i,j)=1-\delta_{ij}$, where $\delta$ is the Kronecker symbol.
A decision procedure is a $B$-POVM $\{M_1,\dots, M_k\}$, where $M_i$ corresponds to the choice $b_i$.  
Then the average loss is the average error probability
\[
\Le_{\ee,\lambda,w}(M)=\sum_{i,j}\lambda_i(1-\delta_{ij})\Tr b_i M_j=\sum_{i\ne j}\lambda_i \Tr b_i M_j.
\]
We can use Remark \ref{rem:payoff_loss} to compute the minimal average error probability 
$\Pi_\lambda^B(b_1,\dots,b_k):=\min_M \Le_{\ee,\lambda,w}(M)$. 
We obtain $\xi_{\ee,\lambda,w'}=\sum_i |i\>\<i|\otimes \lambda_ib_i$, so that the minimal average error probability is
\[
\Pi_\lambda^B(b_1,\dots,b_k)=1-\|\xi_{\lambda,w'}\|_{I\otimes B}=1-\inf_{b\in B}\sup_{1\le i\le k} 2^{D_{max}(\lambda_ib_i\| b)}
\]
For $B=\Se(\Ha)$, the last equality was obtained in \cite{damo}, see also \cite{renner}.

Let us now look at an optimal decision procedure. Let $\{M_i\}$ be a $B$-POVM with decomposition $M=\Lambda\circ\chi_c$
 and let us suppose that $c=\sum_i M_i$ is positive definite. Let $\sigma_i=\chi_c(b_i)$ and $\ee_c=(\Ha,\Se(\Ha),\{\sigma_1,\dots,\sigma_k\})$. Suppose that $\{\Lambda_i\}$ is optimal for $(\ee_c,\lambda,w)$, this is equivalent to the fact that
$\sum_i\lambda_i\sigma_i\Lambda_i=: p$ is a hermitian element that majorizes $\lambda_i\sigma_i$ for all $i$.
By Remark \ref{rem:payoff_loss} and Corollary
 \ref{coro:opt_class_dec}, $\{M_i\}$ is then optimal for $(\ee,\lambda,w)$ if and only if $p\in c^{1/2}Jc^{1/2}$, note that
 $\sigma_i\in c^{1/2}Jc^{1/2}$ for all $i$.

\end{ex}

\begin{ex}[Hypothesis testing] \label{ex:hypo} Let $k=2$ in the previous example, then
 we obtain the hypothesis testing, or discrimination problem, considered at the end of Section 
\ref{sec:basen}. Here we have 
\[
\||0\>\<0|\otimes sb_0 + |1\>\<1|\otimes  tb_1\|_{I_2\otimes B}=\frac 12(\|sb_0-tb_1\|_B+s+t)
\]
for $s,t>0$, so that indeed, $1-\|\xi_{\ee,\lambda,w'}\|_{I_2\otimes B}=\frac12(1-\|\lambda b_0-(1-\lambda)b_1\|_B)$ 
is the 
minimal Bayes error probability.  Let $\{M_0,M_1\}$ be a $B$-POVM such that $c=M_0+M_1$ is positive definite and 
let $\sigma_i=\chi_c(b_i)$. Suppose $\lambda=1/2$ and let $\Lambda_i=c^{-1/2}M_ic^{-1/2}$ be  a POVM which is optimal for 
$(\ee_c,\lambda,w)$, then
$\Lambda_0$ is the projection onto the support of $(\sigma_0-\sigma_1)_+$ and
 $\sum_i\lambda_i\sigma_i\Lambda_i=\frac12((\sigma_0-\sigma_1)_++\sigma_1)$. From the previous example, it is clear that
$\{M_0,M_1\}$ is then an optimal test for $(\ee,\lambda,w)$ if and only if any of (and therefore all of) $(\sigma_0-\sigma_1)_+$,
 $(\sigma_0- \sigma_1)_-$, $|\sigma_0-\sigma_1|$ is an element in $c^{1/2}Jc^{1/2}$.  
\end{ex}

In particular, let $B=\Ce(\Ha,\Ka)$. In this case, the $B$-POVMs are exactly the quantum 1-testers of 
\cite{daria_testers,ziman}, see also \cite{ja}. More precisely, the $B$-POVMs $M=\{M_d, d\in D\}\subset B(\Ka\otimes \Ha)^+$
 satisfy $\sum_d M_d =I\otimes \sigma$ 
for some $\sigma\in \Se(\Ha)$.  Then the decomposition $M=\Lambda\circ \chi_{I\otimes\sigma}$ corresponds to an implementation of the tester $M$ by a triple $(\Ha_A, \rho, \Lambda)$, where $\Ha_A=\supp(\sigma)$,  
$\rho=\chi_{I\otimes \sigma}(\Psi)$  is a pure 
state in $\Se(\Ha\otimes\Ha_A)$ and
\[
\Tr M_d X_{\Phi}=\Tr \Lambda_d((\Phi\otimes id_A)(\rho))
\]
Note that $\sigma=\dim(\Ha)^{-1} I$ is obtained in the case that the  input state $\rho$ is maximally entangled.
By the results of Example \ref{ex:hypo}, we have the following:

\begin{coro} Let $b_i=X_{\Phi_i}$ be Choi matrices of the channels $\Phi_0,\Phi_1: B(\Ha)\to B(\Ka)$. Consider the problem
 of testing the hypothesis $\Phi_0$ against $\Phi_1$, with a priori probability $\lambda\in [0,1]$.
Then there exists an optimal 1-tester with maximally entangled input state if and only if
$\Tr_\Ka |\lambda X_{\Phi_0}-(1-\lambda)X_{\Phi_1}|$  is a multiple of $I_\Ha$.

\end{coro}

\end{document}